\documentclass[journal,onecolumn]{IEEEtran}
\usepackage{amsfonts,amssymb,amsmath,amsthm,bm}
\usepackage{latexsym,url,color,enumerate}
\usepackage{graphicx}
\usepackage{float}
\usepackage[normalem]{ulem}

\newtheorem{theorem}{Theorem}
\newtheorem{proposition}{Proposition}
\newtheorem{lemma}{Lemma}
\newtheorem{corollary}{Corollary}
\newtheorem{definition}{Definition}
\newtheorem{example}{Example}
\newtheorem{problem}{Problem}
\newtheorem{remark}{Remark}
\newtheorem*{lemma1}{Lemma~\ref{IndependentFamily}}
\newtheorem*{lemma2}{Lemma~\ref{IndependenceNumber}}
\newtheorem*{lemma3}{Lemma~\ref{IndependenceNumberStrongProduct}}

\newcommand{\RR}{\mathbb{R}}

\newcommand{\bx}{\mathbf{x}}
\newcommand{\by}{\mathbf{y}}

\newcommand{\dfn}{\triangleq}
\newcommand{\Cl}{\textup{Cl}}

\begin{document}
\title{The $\rho$-Capacity of a Graph}

\author{Sihuang Hu and Ofer Shayevitz
\thanks{This work has been supported by an ERC grant no. 639573, and an ISF grant no. 1367/14. The authors are with the Department of Electrical Engineering--Systems, Tel Aviv University, Tel Aviv, Israel (emails: sihuanghu@post.tau.ac.il, ofersha@eng.tau.ac.il).
This paper was presented in part at the IEEE International Symposium on Information Theory (ISIT), Barcelona, Spain, July 2016.
}
}

\maketitle

\begin{abstract}
Motivated by the problem of zero-error broadcasting, we introduce a new notion of graph capacity, termed $\rho$-capacity, that generalizes the Shannon capacity of a graph. We derive upper and lower bounds on the $\rho$-capacity of arbitrary graphs, and provide a Lov\'asz-type upper bound for regular graphs. We study the behavior of the $\rho$-capacity under two graph operations: the strong product and the disjoint union. Finally, we investigate the connection between the
structure of a graph and its $\rho$-capacity.  
\end{abstract}

\IEEEpeerreviewmaketitle

\section{Introduction}
The zero-error capacity of a discrete memoryless noisy channel was first investigated by Shannon~\cite{Sh1956}. In this setup, a transmitter would like to communicate a message to a receiver through the channel, and the receiver must decode the message without error. The zero-error capacity is the supremum of all achievable communication rates under this constraint, in the limit of multiple channel uses. This problem can be equivalently cast in terms of the
\textit{confusion graph} $G$ associated with the channel. The vertices of the confusion graph are the input symbols, and two vertices are adjacent if the corresponding inputs can result in the same output. Letting $G^n$ denote the $n$th-fold strong product of $G$, which is the confusion graph for $n$ uses of the channel, the zero-error capacity is obtained as the exponential growth rate of $\alpha(G^n)$, the size of  a maximum independent set of $G^n$. The zero-error
capacity is thus often referred to as the \textit{Shannon capacity} of a graph, and is denoted by $C(G)$. Despite the apparent simplicity of the problem, a general characterization of $C(G)$ remains elusive. Lower and upper bounds were obtained by Shannon~\cite{Sh1956}, Lov\'asz~\cite{Lo1979} and Haemers~\cite{Ha1979}. 

Our work is motivated by the more general problem of characterizing the zero-error capacity of the discrete memoryless broadcast channel with two receivers. This problem can be cast in terms of the confusion graphs $(G_1,G_2)$ corresponding to each of the receivers, and the associated capacity region is denoted by $C(G_1,G_2)$. This setup was considered by Weinstein \cite{We2011}, who found the region  $C(G_1,G_2)$ in a few special cases: 
\begin{itemize}
\item Both graphs are disjoint union of cliques. Note that the capacity region in this case can be deduced from more general results by Pinsker~\cite{Pi1978}, Marton~\cite{Ma1977} and Willems~\cite{Wi1990}. 
\item $G_1$ is the complete graph minus a clique, and $G_2$ is either empty or the complement graph of $G_1$. The capacity region in this case is obtained by time sharing between the optimal point-to-point zero-error codes for $G_1$ and $G_2$. 
\end{itemize}

In this paper, we focus on the case where $G_1$ is the empty graph (i.e., the first receiver observes the input noiselessly), but where the graph $G_2$ can be arbitrary. This naturally gives rise to the notion of the \textit{$\rho$-capacity of a graph}. Specifically, the $\rho$-capacity of the graph $G_2$, written as $C_\rho(G_2)$, is the maximal rate that can be conveyed with zero-error to the second receiver, while communicating with the first (noiseless) receiver at a rate of at least
$\rho$. In terms of the graph, the $\rho$-capacity is the exponential growth rate of the maximal number of pairwise non-adjacent subsets of size $2^{\rho n}$ in $G^n$. This notion of capacity generalizes the Shannon capacity of a graph, which is obtained as $C(G) = C_0(G)$. 

Our paper is dedicated to the study of the $\rho$-capacity. In Section~\ref{sec:pre}, we formally define the $\rho$-capacity and explore its relation to the zero-error broadcasting problem. In Section~\ref{sec:rho_capacity}, we provide several upper and lower bounds on the $\rho$-capacity of arbitrary graphs, as well as an upper bound on the $\rho$-capacity of regular graphs that generalizes Lov\'asz's construction~\cite{Lo1979,MRR1978,Sch1979}. In Section~\ref{sec:properties}, we study the behavior of the $\rho$-capacity under two operations on graphs: the strong graph product and the disjoint graph union. 
Some relations between the $\rho$-capacity curve of a graph and its structure are investigated in Section~\ref{sec:structural}. We conclude by briefly discussing a few open problems in Section~\ref{sec:open}.

\section{Preliminaries}\label{sec:pre}
\subsection{Notations and Background}
Let $G=(V,E)$ be a graph with vertex set $V$ and edge set $E$. Two vertices $v_1,v_2$ are called adjacent if there is an edge between $v_1$ and $v_2$, written as $v_1\sim v_2$. An \textit{independent set} in $G$ is a subset of pairwise non-adjacent vertices. A \emph{maximum independent set} is an independent set with the largest possible number of vertices. This number is called the \emph{independence number} of $G$, and denoted by $\alpha(G)$. We write $K_m$ for the complete graph over $m$ vertices. 

Let $G=(V(G),E(G))$ and $H=(V(H),E(H))$ be two graphs. The \emph{strong product} $G\boxtimes H$ of the graphs $G$ and $H$ is a graph such that
\begin{enumerate}
  \item the vertex set of  $G\boxtimes H$ is the Cartesian product $V(G) \times V(H)$; and
  \item any two distinct vertices $(u,u')$ and $(v,v')$ are adjacent in $G\boxtimes H$ if
    $u\sim v$ and $u'=v'$, or $u=v$ and $u'\sim v'$, or $u\sim v$ and $u'\sim v'$.
\end{enumerate}
For graphs $G$ and $H$, we let $G + H$ denote their disjoint union. For a positive integer $n$, we interpret $nG$ as the disjoint union of $n$ copies of $G$. Two graphs $G$ and $H$ are called \emph{isomorphic}, written as $G \cong H$, if there exists a bijection $\varphi$ from $V(G)$ onto $V(H)$ such that any two vertices $u$ and $v$ in $G$ are adjacent if and only if $\varphi(u)$ and $\varphi(v)$ in $H$ are adjacent. Note that the strong product is commutative and associative in the sense that
\begin{align*}
  G_1\boxtimes G_2 &\cong G_2\boxtimes G_1,\\
  (G_1\boxtimes G_2)\boxtimes G_3 &\cong G_1\boxtimes (G_2\boxtimes G_3).
\end{align*}
It is also immediate that the strong product is distributive for the disjoint union:
\begin{align*}
  G_1\boxtimes(G_2+G_3) = G_1\boxtimes G_2+G_1\boxtimes G_3.
\end{align*}
(See~\cite[Section 5.2]{HaImKl2011} for more properties of the strong product.)
The graph $G^{\boxtimes n}$ is defined inductively by $G^{\boxtimes n}= G^{\boxtimes n-1} \boxtimes G$. For simplicity we will write $G^n$ instead of $G^{\boxtimes n}$. 

The \emph{Shannon capacity} of a graph $G$ is defined as the exponential growth rate of the independence number of $G^n$, i.e., 
\begin{align}\label{eq:shannon_cap_def}
 C(G)=\lim_{n\to\infty}\frac{1}{n}\log\alpha(G^n),  
\end{align}
where the limit exists by the superadditivity of $\log\alpha(G^n)$. This quantity also arises as the zero-error capacity in the context of channel coding~\cite{Sh1956}, as we now briefly delineate. 

A (point-to-point) \emph{discrete memoryless channel} consists of a finite input alphabet $\mathcal{X}$, a finite output alphabet $\mathcal{Y}$, and a conditional probability mass function $p(y| x)$, such that $p(y^n| x^n)=\prod_{i=1}^np(y_{i} | x_i)$ when the channel is used $n$ times. A transmitter would like to convey a message\footnote{Throughout the paper we ignore integer issues whenever they are not important.} $w\in [2^{nR}]$ to a receiver over this channel, where the transmitter can set
the input sequence $x^n$ to the channel, and the receiver observes the output sequence $y^n$. To that end, the transmitter and receiver use an \textit{$(n,R)$ code}, which consists of
an encoder $\psi: [2^{nR}]  \rightarrow \mathcal{X}^n$ and a decoder $g: \mathcal{Y}^n \rightarrow [2^{nR}]$. Such a code is said to be \textit{zero-error} if $w$ can always be \textit{uniquely determined} from $y^n$, i.e., $w=g(y^n)$ for any $w$ and any correspondingly feasible output sequence $y^n$. We say that the communication rate $R$ is \emph{achievable} if an $(n,R)$ zero-error code exists\footnote{Note that since a concatenation of two zero-error codes is a
zero-error code, $R$ is achievable for arbitrarily large $n$.} for some $n$. The \emph{zero-error capacity} of the channel is defined to be the supremum of all achievable rates.

The channel $p(y|x)$ can be associated with a \emph{confusion graph} $G$, whose vertex set is the input alphabet $\mathcal{X}$, and whose edge set consists of all input pairs $(x,x')$ that can lead to the same output, i.e., for which both $p(y|x)>0$ and $p(y|x') >0$ for some $y\in\mathcal{Y}$. It is easy to verify that $G^n$ is the confusion graph associated with the product channel $p(y^n| x^n)$. It is well known and easy to check that the zero-error capacity of the channel is
equal to $C(G)$, the Shannon capacity of its confusion graph $G$.

\subsection{$\rho$-Capacity}
In this subsection we introduce the $\rho$-capacity of a graph, which is a generalization of the Shannon capacity. We begin by generalizing the notion of an independent set of a graph. Let $G=(V,E)$ be a graph with vertex set $V$ and edge set $E$.
Two disjoint subsets of vertices $V_1,V_2$ are called \emph{adjacent} if there exist vertices $v_1\in V_1$ and $v_2\in V_2$ such that $v_1\sim v_2$. Let $\mathcal{F}=\{V_i: 1\le i\le l\}$ be a family of disjoint subsets $V_i\subset V$. If they are pairwise non-adjacent, then we call $\mathcal{F}$ an \emph{independent family} of $G$. Moreover,  if each $V_i$ is of size not
less than $k$, then we say $\mathcal{F}$ is a \emph{$k$-independent family} of $G$. We write $|\mathcal{F}|$ for the number of subsets in $\mathcal{F}$. A \emph{maximum $k$-independent family} is a $k$-independent family with the largest possible number of subsets. This number is called the \emph{$k$-independence number} of $G$, and denoted by $\alpha_k(G)$. In particular, we have $\alpha_1(G)=\alpha(G)$, where $\alpha(G)$ is the independence number of $G$. 
\begin{example}{\rm
  Let $C_5$ be the pentagon graph, whose vertex set is $\{1,2,3,4,5\}$ and edge set is $\{12,23,34,45,51\}$. Then $\mathcal{F}=\{ \{1,2\}, \{4\}\}$ is an independent family of $C_5$, and it is easy to verify that $\alpha_2(C_5)=1$.
}
\end{example}

Here is a simple property of the $k$-independence number.
\begin{lemma} \label{StrongProductIndependenceNumberInequality}
  Let $G$ and $H$ be two graphs, and let $k_1$ and $k_2$ be two positive integers. Then
  $$\alpha_{k_1k_2}(G\boxtimes H) \ge \alpha_{k_1}(G)\cdot\alpha_{k_2}(H).$$
\end{lemma}
\begin{proof}
Suppose that $\{V_i: 1\le i\le l\}$ is a $k_1$-independent family of $G$, and $\{U_j: 1\le j\le l'\}$ is a $k_2$-independent family of $H$. Then their Cartesian product $\{V_i\times U_j: 1\le i\le l, 1 \le j\le l'\}$ is a $k_1k_2$-independent family of $G\boxtimes H$, and the result follows.
\end{proof}

We now define the $\rho$-capacity of a graph $G$ to be the exponential growth rate of the $2^{\rho n}$-independence number of $G^n$, which generalizes the expression~\eqref{eq:shannon_cap_def} for the Shannon capacity. 
\begin{definition}
Let $G$ be a graph with $m$ vertices. Then for any $0\le\rho\le\log{m}$, the \emph{$\rho$-capacity} of $G$ is defined to be 
  \begin{align}\label{eq:rho_cap_def}
    C_{\rho}(G) = \lim_{n\to\infty}\frac{1}{n}\log \alpha_{2^{\rho n}}(G^n).
  \end{align}
In particular, the $\rho$-capacity for $\rho=0$ is equal to the Shannon capacity of the graph, i.e., $C_0(G)= C(G)$.
\end{definition}

Note that the existence of the limit~\eqref{eq:rho_cap_def} follows from the superadditivity of $\log\alpha_{2^{\rho n}}(G^n)$, namely
$$
\log\alpha_{2^{\rho(n+n')}}(G^{n+n'}) 
\ge \log\alpha_{2^{\rho n}}(G^n) + \log\alpha_{2^{\rho n'}}(G^{n'}), 
$$
which is guaranteed by Lemma~\ref{StrongProductIndependenceNumberInequality}. In particular, it also holds that $\displaystyle{C_\rho(G) = \sup\{\tfrac{1}{n}\log \alpha_{2^{\rho n}}(G^n): n=1,2,\dots\}}$.

\subsection{A Zero-Error Broadcasting Formulation} 
In this subsection, we show how the $\rho$-capacity arises naturally in the context of zero-error broadcasting. A (two-user) \emph{discrete memoryless broadcast channel} consists of a finite input alphabet $\mathcal{X}$, two finite output alphabets $\mathcal{Y}_1$ and $\mathcal{Y}_2$, and a conditional probability mass function $p(y_1,y_2|x)$, such that $p(y_1^n,y_2^n|x^n)=\prod_{i=1}^np(y_{1i},y_{2i}|x_i)$ when the channel is used $n$ times. A transmitter would like to convey two messages $w_1\in [2^{nR_1}]$ and $w_2\in [2^{nR_2}]$ to two receivers over this channel, where the transmitter can set the input sequence $x^n$ to the channel, and the receivers observe their respective output sequences $y_1^n$ and $y_2^n$. To that end, the transmitter and receivers use an \textit{$(n,R_1,R_2)$ code}, which consists of
an encoder $\psi: [2^{nR_1}] \times [2^{nR_2}] \rightarrow \mathcal{X}^n$, and two decoders $g_1: \mathcal{Y}_1^n \rightarrow [2^{nR_1}]$ and $g_2: \mathcal{Y}_2^n \rightarrow [2^{nR_2}]$. Such a code is said to be \textit{zero-error} if $w_1$ and $w_2$ can always be \textit{uniquely determined} from $y_1^n$ and $y_2^n$, i.e., $w_1=g_1(y_1^n)$ and $w_2=g_2(y_2^n)$ for any $w_1,w_2$ and any correspondingly feasible pair of output sequences. We say that the communication rates $(R_1,R_2)$ are \emph{achievable} if an $(n,R_1,R_2)$ zero-error code exists for some $n$. The \emph{zero-error capacity region} of the broadcast channel is the closure of the set of all achievable rates.

Similarly to the case of the broadcast capacity under a vanishing-error criterion, it is easy to observe the following result.
\begin{proposition}\label{prop:depends}
  The zero-error capacity region of a broadcast channel depends only on the conditional marginal
  distributions $p(y_1|x)$ and $p(y_2|x)$.
\end{proposition}
Let $G_1=(\mathcal{X},E_1)$ and $G_2=(\mathcal{X},E_2)$ be the confusion graphs associated with the channels $p(y_1|x)$ and $p(y_2|x)$ respectively. Then Proposition~\ref{prop:depends} implies a simple corollary.
\begin{corollary}
  The zero-error capacity region of a broadcast channel depends only on the confusion graphs 
  $G_1$ and $G_2$.  
\end{corollary}

Following this, we write $C(G_1,G_2)$ for the zero-error capacity region of a broadcast channel with confusion graphs $G_1$ and $G_2$. We now show that the $\rho$-capacity of $G$ is the maximal rate that can be conveyed under zero-error to a noisy receiver with confusion graph $G$, while at the same time communicating with a noiseless receiver (i.e., having an empty confusion graph) at a rate of at least $\rho$. 
\begin{proposition}\label{prop:rho-capacityForBroadcast}
  Let $G$ be a graph over $m$ vertices. Then for any $0\le\rho\le\log{m}$, we have
\begin{align*}
  C_\rho(G) = \sup\,\{R : (\rho,R)\in C(\overline{K}_m,G)\}
\end{align*}
where $\overline{K}_m$ is the empty graph over $m$ vertices. 
\end{proposition}
\begin{proof}
Let $\mathcal{F}$ be a $2^{\rho n}$-independent family of $G^n$, and set $R=\tfrac{1}{n}\log|\mathcal{F}|$. Then $\mathcal{F}$ induces an $(n,\rho,R)$ zero-error code for the broadcast setup associated with the definition of $\rho$-capacity, i.e., where the first receiver is noiseless and the second receiver has confusion graph $G$. 
The $(n,\rho,R)$ zero-error code is constructed using \textit{superposition coding}: the transmitter chooses a subset of $\mathcal{F}$ for the second receiver, and chooses a vertex inside that subset for the first receiver, which is then transmitted. Clearly, the second receiver can always distinguish between the subsets of $\mathcal{F}$ (hence decode its message with zero-error), whereas the first receiver can decode both messages with zero-error. 
Therefore $(\rho, C_{\rho}(G))\in C(\overline{K}_m,G)$, and hence $C_\rho(G) \le \sup\,\{R : (\rho,R)\in C(\overline{K}_m,G)\}$.

Conversely, suppose that the rate pair $(\rho,R)$ is achievable, i.e., there exists an $(n,\rho,R)$ zero-error code for some $n$. Consider the subsets of codewords 
  obtained by fixing the second receiver's message and going over all the messages of the first receiver. 
  All these subsets are of size $2^{\rho n}$, and since the second receiver must decode with zero-error regardless of the first receiver's message, any pair of these subsets must be non-adjacent in $G^n$. This naturally induces a $2^{\rho n}$-independent family of $G^n$ whose number of subsets is equal to $2^{n R}$. Therefore $R\le C_{\rho}(G)$, and hence $C_\rho(G) \ge \sup\,\{R : (\rho,R)\in C(\overline{K}_m,G)\}$. This concludes our proof.
\end{proof}

The following proposition shows how the $\rho$-capacity can be used to provide a partial characterization of the zero-error broadcast capacity region. 
\begin{proposition}
  Let $G_1=(\mathcal{X},E_1)$ and $G_2=(\mathcal{X},E_2)$ be two graphs. Let $\mathfrak{C}$ be the convex hull of the closure of all 
      rate pairs $(R_1,R_2)$ satisfying
      \begin{align*}
        R_1 &\le \rho/s,\\ R_2 &\le  C_{\rho}(H)/s,
      \end{align*}
      where $H$ is the induced subgraph of $G_2^s$ associated with some independent set $A$ of 
      $G_1^s$ for some positive integer $s$, and $0\le\rho\le\log{|A|}$.
  Then $\mathfrak{C}\subseteq C(G_1,G_2)$. Moreover, if $E_1\subseteq E_2$ then $\mathfrak{C}=C(G_1,G_2)$. 
\end{proposition}
\begin{proof}
  Fix a $\rho\in[0,\log{|A|}]$. Then there exists a sequence of $2^{\rho n}$-independent families $\mathcal{F}_n$ of $H^n$ such that $\lim_{n\to\infty}\frac{1}{n}\log{|\mathcal{F}_n|}= C_{\rho}(H)$. 
  Using superposition coding, we can see that $\mathcal{F}_n$ is an 
  $(sn,\rho,\frac{1}{n}\log|\mathcal{F}_n|)$ zero-error code. Hence $\mathfrak{C}\subseteq C(G_1,G_2)$. The second statement can be proved similarly as in Proposition~\ref{prop:rho-capacityForBroadcast}.
\end{proof}

\subsection{Simple Properties of the $\rho$-Capacity}
\begin{proposition}
  Let $G$ be a graph with $m$ vertices. The following properties of its $\rho$-capacity are easily observed: 
\begin{enumerate}
\item $C_0(G)= C(G)$, i.e., the $\rho$-capacity for $\rho=0$ is equal to the Shannon capacity of the graph. 
\item $C_{\log{m}}(G) = 0$. 
\item $C_{\rho}(G)$ is monotonically non-increasing in $\rho$ on $[0,\log{m}]$ (by definition).  
\item $C_{\rho}(G)$ is concave in $\rho$ on $[0,\log{m}]$ (by time sharing).
\end{enumerate}  
\end{proposition}

We now define three quantities related to the $\rho$-capacity, which  will be of interest in the sequel. Let $G$ be a graph over $m$ vertices. We write $\rho^*(G)$ for the maximal $\rho\in[0,\log{m}]$ such that $C_{\rho}(G)=C(G)$, and $\rho_{*}(G)$ for the minimal $\rho\in[0,\log{m}]$ such that $C_{\rho}(G)=\log{m}-\rho$. We refer to  $\rho^*(G)$ as the \emph{free-lunch point} of $G$, and to $\rho_{*}(G)$ as the \emph{packing point} of $G$. The \emph{concave conjugate} of $C_\rho(G)$ is defined as
$$ C_{\star}(G,\gamma)\triangleq\inf_{\rho\in[0,\log{m}]}\gamma\rho-C_{\rho}(G) \qquad\text{for } \gamma\in[-1,0].$$

Here are two simple bounds on the $\rho$-capacity in terms of the Shannon capacity $C(G)$. The lower bound follows by time sharing (concavity), and the upper bound follows from the definition of an independent family.
\begin{proposition}\label{lem:trivial_bound_rho}
  Let $G$ be a graph with $m$ vertices. Then, for $0\le\rho\le\log{m}$, we have
  \begin{align}\label{NaiveBounds}
  \frac{C(G)}{\log{m}}(\log{m}-\rho) \le C_{\rho}(G)\le\min\{C(G), \log{m}-\rho\}.
\end{align}
\end{proposition}

\begin{example}{\rm 
    Let $G$ be a disjoint union of two cliques, each of size $\tfrac{m}{2}$. It is easy to see that $C(G)=1$, and that $\alpha_{\frac{m}{2}}(G)=2$. Hence, $C_\rho(G) \geq 1$ for any $\rho\in[0,\log{\tfrac{m}{2}}]$, and by concavity also $C_\rho(G) \geq \log{m}-\rho$ for any $\rho\in[\log{\tfrac{m}{2}}, \log{m}]$. Hence the upper bound from Proposition~\ref{lem:trivial_bound_rho} is tight in this case. In particular, the free-lunch point and the packing point coincide,
    $\rho^*(G)=\rho_{*}(G)=\log{\tfrac{m}{2}}$. The concave conjugate is given by $C_{\star}(G,\gamma) = \gamma\rho_*(G) - C_0(G) = \gamma\log{\tfrac{m}{2}} - 1$.  
}
\end{example}

\begin{example}{\rm
    \cite[Theorem 7]{We2011}
  Let $G$ be the complete graph on $m$ vertices minus a clique on $d$ vertices. Then
  $$C_{\rho}(G)=\log{d}-\frac{\log{d}}{\log{m}}\,\rho.$$
  This meets the lower bound of Proposition~\ref{lem:trivial_bound_rho}. In particular, the free-lunch point $\rho^*(G)=0$ and the packing point $\rho_{*}(G)=\log{m}$. The concave conjugate is given by $C_{\star}(G,\gamma) = \min\left\{-\log{d}, \gamma\log{m}\right\}$.  
  \label{CliqueMinusClique}
}
\end{example}

\section{Bounds on the $\rho$-Capacity}\label{sec:rho_capacity}
In this section, we give three types of bounds on the $\rho$-capacity of a graph. The first bound is trivially derived from the capacity region of the degraded broadcast channel under the vanishing-error criterion. The second is based on the distribution of independent families and clique covers, via an explicit expression for the $\rho$-capacity of a disjoint union of cliques. The third generalizes Lov\'asz's $\vartheta$-function upper bound for the Shannon capacity. 

\subsection{An Information-Theoretic Upper Bound}
The random variables $X,Y,Z$ are said to form a \emph{Markov chain} in that order, denoted by $X-Y-Z$, if their joint probability mass function can be written as $p(x,y,z)=p(x)p(y|x)p(z|y).$

\begin{theorem} \label{UpperBoundFromVanishingError}
  The $\rho$-capacity of a graph $G$ satisfies
  $$C_\rho(G) \leq \min_{p(y|x)}\max_{\substack{U-X-Y\\ H(X|U) \ge \rho\\}} I(U;Y)$$
where the min is taken over all possible point-to-point channels $p(y|x)$ associated with a confusion graph $G$, and the random variable $U$ has cardinality bounded by $|\mathcal{U}|\le\min\{|\mathcal{X}|,|\mathcal{Y}|\}$.
\end{theorem}
\begin{proof}
  Consider a broadcast channel where the first receiver sees a noiseless channel, i.e., observes the input $x$, and the second receiver sees the input $x$ through a noisy channel $p(y|x)$. The capacity region for this broadcast channel under the vanishing-error criterion is the convex hull of the closure of all $(R_1,R_2)$ satisfying
  \begin{align}\label{capacity_region_vanishing}
    \begin{split}
      R_1 &\le H(X|U),\\
      R_2 &\le I(U;Y)
    \end{split}
  \end{align}
  for some Markov chain $U-X-Y$, where the auxiliary random variable $U$ has cardinality bounded by $|\mathcal{U}|\le\min\{|\mathcal{X}|,|\mathcal{Y}|\}$ (see \cite[Theorem 15.6.2]{CoTh2006}). In particular, if $p(y|x)$ has confusion graph $G$, then the above region contains the zero-error capacity region $C(\overline{K}_m,G)$ where $m=|\mathcal{X}|$. The result now follows from Proposition~\ref{prop:rho-capacityForBroadcast}. 
\end{proof}

\subsection{Upper and Lower Bounds Based on Disjoint Union of Cliques}

Recall that the R\'{e}nyi entropy of order $\beta$, where $\beta \geq 0$ and $\beta \neq 1$,
is defined as
$$
H_{\beta}(P)=\frac{1}{1-\beta}\log\sum_{i=1}^s p_i^{\beta},
$$
where $P=\{p_1,\ldots,p_s\}$ is a probability distribution.
The limiting value of $H_\beta$ as $\beta \rightarrow 1$ is the Shannon entropy
$H_1(P)=H(P)$. 
Let $\mathcal{F}$ be a family of disjoint subsets of sizes $\{m_1,\ldots,m_s\}$. We define $M_\mathcal{F}\dfn\sum_{i=1}^sm_i$, and $Q_\mathcal{F}$ to be the distribution induced by the family, namely the distribution $\left(m_1/M_\mathcal{F},\ldots, m_s/M_\mathcal{F}\right)$. 

The Shannon capacity satisfies~\cite{Sh1956}
\begin{align}\label{sandwich_bound}
  \log{\alpha(G)} \le C(G) \le \log{{\rm cc}(G)}
\end{align}
where $\alpha(G)$ is the independence number of $G$ and cc($G$) is the vertex clique covering number of $G$. A \emph{vertex clique covering} of a graph $G$ is set of cliques such that every vertex of $G$ is a member of exactly one clique. A minimum clique covering is a clique covering of minimum size, and the \emph{clique covering number} cc($G$) is the size of a minimum clique covering.
Now we describe a natural generalization of this bound to the $\rho$-capacity, which also include the bound~\eqref{sandwich_bound} as a special case when $\rho=0$. Note that the derivations in the proofs of Theorems~\ref{thm:bounds}--\ref{rho_capacity_disjoint_union_of_cliques} are incidentally very similar to those of~\cite[Chapter 5]{jelinek}, where Jelinek calculated the error exponents in source coding.

\begin{theorem}\label{thm:bounds}
 Let $G$ be a graph with $m$ vertices. Suppose that $\mathcal{F}_1$ is an independent family of $G$ and $\mathcal{F}_2$ is a vertex clique cover of $G$.
Then
  \begin{align}\label{lower_bound_rho}
    C_{\rho}(G) \ge \inf_{\beta\in[0,1]}(1-\beta)H_{\beta}(Q_{\mathcal{F}_1})+\beta(\log{M_{\mathcal{F}_1}}-\rho)
  \end{align}
for $0\leq\rho\leq \log{M_{\mathcal{F}_1}}$, and 
  \begin{align}\label{upper_bound_rho}
    C_{\rho}(G) \le \inf_{\beta\in[0,1]}(1-\beta)H_{\beta}(Q_{\mathcal{F}_2})+\beta(\log{m}-\rho)
  \end{align}
  for $0\le\rho\le\log{m}$. Moreover, if $G$ is a disjoint union of cliques, then both bounds coincide (and are hence tight). The minimizing $\beta$ is given in~\eqref{rho_expression} and~\eqref{beta-expression}. 
\end{theorem}
\begin{proof}
  We prove the bounds are tight for $G=K_{m_1}+K_{m_2}+\cdots+K_{m_s}$ that is a disjoint union of cliques.\footnote{Note that the bounds for a disjoint union of cliques appear implicitly in~\cite{We2011, Pi1978, Ma1977, Wi1990}. Here we provide the exact analytical expression.} The lower bound for a general graph follows since given an independent family, we can consider the associated induced subgraph, and add edges to create a disjoint union of cliques (hence decrease the  $\rho$-capacity). The upper bound for a general graph will follow by noting that given a vertex clique cover, we can remove edges (hence increase the $\rho$-capacity) to create a disjoint union of cliques. 

The case that $m_i,1\le i\le s$ are all equal is easy to prove, thus we can assume that $s\ge2$ and $m_i,1\le i\le s$ are not all equal. Set $\mathcal{F}=\{V(K_{m_1}),\ldots,V(K_{m_s})\}$ and $m=m_1+\dots+m_s$. Now we have 
\begin{align}\label{G^n}
   G^n = (K_{m_1}+K_{m_2}+\cdots+K_{m_s})^{n} \cong \sum_{i_1+\cdots+i_s=n}{n\choose i_1,i_2,\ldots,i_s}K_{m_1^{i_1}\cdots m_s^{i_s}},
 \end{align}
where the sum is taken over all combinations of nonnegative integer indices $i_1$ through $i_s$ such that the sum of all $i_j$ is $n$. 
From~\eqref{G^n} we see $G^n$ is also a disjoint union of cliques. Let us write $G^n$ as a disjoint union of small and large cliques, i.e.,  $G^n=G_1+G_2$ where
\begin{align*}
  G_1 &= \sum_{\substack{i_1+\cdots+i_s=n\\m_1^{i_1}\cdots m_s^{i_s}< 2^{\rho n}}}{n\choose i_1,i_2,\ldots,i_s}K_{m_1^{i_1}\cdots m_s^{i_2}},\\
  G_2 &= \sum_{\substack{i_1+\cdots+i_s=n\\m_1^{i_1}\cdots m_s^{i_s}\ge2^{\rho n} }}{n\choose i_1,i_2,\ldots,i_s}K_{m_1^{i_1}\cdots m_s^{i_2}}.
\end{align*}
It is easy to see that $\alpha_{2^{\rho n}}(G_2)\le \alpha_{2^{\rho n}}(G^n) \le \alpha_{2^{\rho n}}(G_1)+\alpha_{2^{\rho n}}(G_2)$, and
\begin{align*}
  \alpha_{2^{\rho n}}(G_1)\le A(n) &\dfn \frac{1}{2^{\rho n}}\sum_{\substack{i_1+\cdots+i_s=n\\m_1^{i_1}\cdots m_s^{i_s}\le 2^{\rho n}}}{n\choose i_1,i_2,\ldots,i_s}m_1^{i_1}\cdots m_s^{i_2}=\frac{|V(G_1)|}{2^{\rho n}},\\
 \alpha_{2^{\rho n}}(G_2) = B(n) &\dfn \sum_{\substack{i_1+\cdots+i_s=n\\m_1^{i_1}\cdots m_s^{i_s}\ge2^{\rho n} }}{n\choose i_1,i_2,\ldots,i_s}.
\end{align*}
Hence
\begin{align}
B(n)\le \alpha_{2^{\rho n}}(G^n) \le A(n)+B(n).
  \label{sandwich_bound_G^n}
\end{align}

Suppose that
\begin{align}\label{rho_interval}
  \frac{1}{s}\sum_{i=1}^s \log{m_i}\le\rho\le\frac{1}{m}\sum_{i=1}^s m_i\log{m_i}.
\end{align}
By~\eqref{sandwich_bound_G^n} and Lemma~\ref{optimization_problems} of the Appendix, we have 
\begin{align}\label{rho-function-1}
  C_{\rho}(G)
  =\lim_{n\rightarrow\infty}\frac{1}{n}\log{A(n)}
  =\lim_{n\rightarrow\infty}\frac{1}{n}\log{B(n)}
  =\log{\left(\sum_{i=1}^s m_i^{\beta}\right)}-\beta\rho
  =(1-\beta)H_{\beta}(Q_{\mathcal{F}})+\beta(\log{M_{\mathcal{F}}}-\rho)
\end{align}
where $\beta\in[0,1]$ is the unique solution such that
 \begin{align}\label{rho_expression}
   \rho= \left(\sum_{i=1}^s m_i^{\beta}\log{m_i}\right)\Big{/}\sum_{i=1}^s m_i^{\beta}.
 \end{align}
Note that: if $\rho=\frac{1}{s}\sum_{i=1}^s \log{m_i}$ associated with $\beta=0$ then $C_{\rho}(G)=C(G)=\log{s}$; and if $\rho=\frac{1}{m} \sum_{i=1}^s m_i\log{m_i}$ associated with $\beta=1$ then $C_{\rho}(G)=\log{m}-\rho$. Therefore
\begin{align}\label{rho-capacity-outrange}
    \begin{split}
    C_{\rho}(G)=\begin{cases}
      \log{s} & \text{if }0\le\rho\le\frac{1}{s}\sum_{i=1}^s \log{m_i}\\
      \log{m}-\rho & \text{if }\frac{1}{m} \sum_{i=1}^s m_i\log{m_i}\le\rho\le\log{m},
    \end{cases}
  \end{split}
  \end{align}
  where the first equality follows from the monotonically decreasing property of $C_{\rho}(G)$, and the second follows by time sharing. Via direct computation, we can verify
  \begin{align}\label{rho-function-2}
    \begin{split}
    C_{\rho}(G) 
    &=\inf_{\beta\in[0,1]}\log{\left(\sum_{i=1}^s m_i^{\beta}\right)}-\beta\rho \\
    &= \inf_{\beta\in[0,1]}(1-\beta)H_{\beta}(Q_{\mathcal{F}})+\beta(\log{M_{\mathcal{F}}}-\rho)
  \end{split}
  \end{align}
  for $0\le\rho\le\log{M_{\mathcal{F}}}$. Here
  \begin{align}\label{beta-expression}
    \begin{split}
    \beta^* &= \underset{\beta\in[0,1]}{\arg\inf}\ (1-\beta)H_{\beta}(Q_{\mathcal{F}})+\beta(\log{M_{\mathcal{F}}}-\rho)\\
  &=\begin{cases}
    0 & \text{if }0\le\rho\le\frac{1}{s}\sum_{i=1}^s \log{m_i}\\
    1 & \text{if }\frac{1}{m} \sum_{i=1}^s m_i\log{m_i}\le\rho\le\log{M_{\mathcal{F}}},
  \end{cases}
\end{split}
\end{align}
and $\beta^*\in[0,1]$ is the unique solution satisfying~\eqref{rho_expression} when $\rho$ satisfies~\eqref{rho_interval}. 
\end{proof}

\begin{remark}{\rm
Note that for $\rho=0$, the bound~\eqref{lower_bound_rho} yields $C(G)\ge\log{\alpha(G)}$. Moreover, if we pick $\beta=0$ in~\eqref{upper_bound_rho},
then it follows that $C_{\rho}(G)\le\log{{\rm cc}(G)}$, and
if we pick $\beta=1$ in~\eqref{upper_bound_rho},
then it follows that $C_{\rho}(G)\le\log{m}-\rho$.
}
\end{remark}

In the following theorem we provide an alternative characterization for the $\rho$-capacity of a disjoint union of cliques, via its concave conjugate. We also explicitly find the associated free-lunch point and packing point. 
\begin{theorem}\label{rho_capacity_disjoint_union_of_cliques}
  Let $G=K_{m_1}+K_{m_2}+\cdots+K_{m_s}$ be a disjoint union of cliques and $m=m_1+\dots+m_s$. Suppose that $s\ge2$ and $m_i,1\le i\le s$ are not all equal.
  \begin{enumerate}
    \item  The concave conjugate of the $\rho$-capacity is given by  
  \begin{align*}
     C_{\star}(G,\gamma)&=-\log{\sum_{i=1}^s m_i^{-\gamma}} \quad \text{for } -1\le\gamma\le0.
  \end{align*}
    \item 
      The $\rho$-capacity $C_{\rho}(G)$ is differentiable on $[0,\log{m}]$ and 
      $$C_{\rho}'(G)=-\,\underset{\beta\in[0,1]}{\arg\inf}(1-\beta)H_{\beta}(Q_{\mathcal{F}})+\beta(\log{M_{\mathcal{F}}}-\rho).$$
    \item The free-lunch point $\rho^*(G)=\frac{1}{s}\sum_{i=1}^s \log{m_i}$.
    \item The packing point $\rho_*(G)=\frac{1}{m} \sum_{i=1}^s m_i\log{m_i}$.
  \end{enumerate}
\end{theorem}

\begin{proof}
Let $g(\gamma) = -\log{\sum_{i=1}^s m_i^{-\gamma}}$ for $\gamma\in[-1,0]$. Define $g_{\star}:[0,\log{m}]\rightarrow \RR$ to be the concave conjugate of $g$, namely $g_{\star}(\rho)=\underset{\gamma\in[-1,0]}{\inf}\rho\gamma-g(\gamma)$. 
  By~\eqref{rho-function-2} we have
  \begin{align*} 
    C_{\rho}(G) 
    &=\inf_{\beta\in[0,1]}\log{\left(\sum_{i=1}^s m_i^{\beta}\right)}-\beta\rho \\
    &= \inf_{\gamma\in[-1,0]}\rho\gamma+\log{\sum_{i=1}^s m_i^{-\gamma}}\\
    &= g_{\star}(\rho).
  \end{align*}
  By the Fenchel--Moreau Theorem~\cite[Exercise 3.39]{BoVa2004} and~\cite[Theorem 4.1.1]{HiJeLe1993}, we have $ C_{\star}(G,\gamma)=g(\gamma)$ for $-1\le\gamma\le0$, and
  \begin{align*}
  C_{\rho}'(G) &=g_{\star}'(\rho) 
   =\underset{\gamma\in[-1,0]}{\arg\inf}\,\gamma\rho-g(\gamma)
  =-\,\underset{\beta\in[0,1]}{\arg\inf}(1-\beta)H_{\beta}(Q_{\mathcal{F}})+\beta(\log{M_{\mathcal{F}}}-\rho).
\end{align*}
This proves 1) and 2). Thus, we have $C_{\rho}'(G)<0$ for $\rho>\frac{1}{s}\sum_{i=1}^s \log{m_i}$ and $C_{\rho}'(G)>-1$ for $\rho<\frac{1}{m} \sum_{i=1}^s m_i\log{m_i}$, and then 3) and 4) follow.
\end{proof}

Next, we provide bounds on the free-lunch point and packing point of general graphs. 

\begin{corollary}\label{free-lunch-rate-and-packing-rate}
  Let $G$ be a graph with $m$ vertices.
  \begin{enumerate}
    \item 
  Suppose that $G$ has $s$ connected components of sizes $m_1,\ldots,m_s$. Then the packing point satisfies $\rho_{*}(G)\le \frac{1}{m} \sum_{i=1}^s m_i\log{m_i}$.
\item
Let $t$ be a positive integer, and let $\mathcal{F}=\{V_1,\dots,V_n\}$ be an independent family of $G^t$. If $C_{0}(G)=(\log{|\mathcal{F}|})/t$ (the Shannon capacity is finitely attained), then the free-lunch point satisfies $\rho^*(G)\ge \frac{1}{tn}\sum_{i=1}^{n}\log{|V_i|}$.
  \end{enumerate}
\end{corollary}
\begin{proof}
  \begin{enumerate}
    \item 
  Note that these connected components trivially form an independent family of $G$. From~\eqref{lower_bound_rho} and~\eqref{rho-capacity-outrange}, we have $C_{\rho}(G)\ge\log{m}-\rho$ for $\frac{1}{m} \sum_{i=1}^s m_i\log{m_i} \le\rho\le\log{m}$. This meets the upper bound of Proposition~\ref{lem:trivial_bound_rho}, and is hence tight.
\item
  By~\eqref{lower_bound_rho} and~\eqref{rho-capacity-outrange} we have $C_{\rho}(G)=C_{0}(G)=(\log{|\mathcal{F}|})/t$ for $0\le\rho\le\frac{1}{tn}\sum_{i=1}^{n}\log{|V_i|}$. Hence $\rho^*(G)\ge\frac{1}{tn}\sum_{i=1}^{n}\log{|V_i|}$.
  \end{enumerate}
\end{proof}

We will later prove (in Theorem~\ref{PackingPoint}) that the upper bound on $\rho_*(G)$ in Corollary~\ref{free-lunch-rate-and-packing-rate} is in fact always attained.

%
\begin{example}{\rm 
Let $G$ be the graph with vertex set $\{1,2,3,4,5\}$ and edge set $\{12,13,23,34,45,15\}$ (pentagon with an extra edge). It is easy to see that $C_{0}(G)=1$. Using the independent family $\mathcal{F}=\{\{2\},\{4,5\}\}$ in  Corollary~\ref{free-lunch-rate-and-packing-rate}, we see that $G$ has a nontrivial free-lunch point $\rho^*(G)\ge 1/2 > 0$.  
}
\end{example}

\begin{example}{\rm
    Let $K_{m,n}\ (m\ge n\ge 2)$ be the complete bipartite graph whose vertices can be partitioned into two subsets $V_1=\{1,2,\dots,m\}$ and $V_2=\{m+1,\dots,m+n\}$ such that no edge has both endpoints in the same subset, and every possible edge that could connect vertices in different subsets is part of the graph. Let $G$ be obtained from $K_{m,n}$ by deleting the edges that connect $m+1$ and $2,3,\dots,m$. Using the independent family $\mathcal{F}=\{\{1,m+1\},\{2\},\{3\},\dots,\{m\}\}$ in Corollary~\ref{free-lunch-rate-and-packing-rate}
     we see that $G$ has a nontrivial free-lunch point $\rho^*(G)\ge1/m>0$.
}
\end{example}

\subsection{A Lov\'asz-Type Upper Bound for Regular Graphs}\label{sec:bound}
In this section, we give an upper bound for the $\rho$-capacity of regular graphs.
Our approach follows the technique developed in~\cite{MRR1978,Sch1979}, which generalized 
Lov\'asz's brilliant idea~\cite{Lo1979}.

Let $G=(V,E)$ be a graph with $m$ vertices and $e\ge1$ edges.
The graph $G$ is said to be \emph{regular} if the number of edges containing a given vertex
$v$ is a constant $r$, independent of $v$, called the \emph{degree} of $G$.
Let $B$ be the adjacency matrix of $G$, and $\mu$ be its smallest eigenvalue. 
It is well-known that $\mu\le -1$ if $G$ has at least one edge.
Set $A = I+|\mu|^{-1}B$.
For two matrices $M$ and $N$, 
we use $M\otimes N$ to denote their Kronecker product. The matrix $M^{\otimes n}$ is
defined inductively by $M^{\otimes n}\dfn M^{\otimes n-1}\otimes M$.
Define
\begin{align*}
\lambda(M)\dfn \inf\{\bx^TM\bx: \sum \bx(i) =1\}.  
\end{align*}

\begin{lemma}[\cite{Lo1979,MRR1978}]\label{lemma:lambda}
  \ \\[-5mm]
  \begin{enumerate}
    \item
      $\lambda(A^{\otimes n})=\lambda(A)^n$.
    \item
      $ C(G) \le \log\lambda(A)^{-1}=\log\frac{m|\mu|}{r+|\mu|}.$
  \end{enumerate}
\end{lemma}
\begin{proof}
  See Section II of~\cite{MRR1978}. 
\end{proof}

Set $V=\{v_1,\dots,v_m\}$. 
Let $\mathcal{F}$ be a $k$-independent family in $G$. 
We can assume that each subset in $\mathcal{F}$ contains exactly $k$ vertices; otherwise we can form another
$k$-independent family $\mathcal{F}\,'$ by choosing exactly $k$ vertices from each subset.
For a vertex $v\in V$, we say that $v\in \mathcal{F}$ if $v$ is contained in some subset of $\mathcal{F}$. Recall that we write $|\mathcal{F}|$ for the number of subsets in $\mathcal{F}$.
We now define a length $m$ vector $\by$ by
$\by(i)=1/(k|\mathcal{F}|)$ if $v_i\in \mathcal{F}$; otherwise $\by(i)=0$.
Then $\sum_i \by(i)=1$ and
\begin{align}\label{bound}
  \lambda(A) \le \by^TA\by = \frac{1}{(k|\mathcal{F}|)^2}(k|\mathcal{F}|+\sum_{\substack{v_i,v_j\in \mathcal{F}\\v_i\sim v_j}}A(i,j)).
\end{align}
Let $\mathcal{F}_n$ be a maximum $ 2^{\rho n}$-independent family of graph $G^n$.
Then by Lemma~\ref{lemma:lambda} and~\eqref{bound} we obtain
\begin{align}\label{bound2}
  \lambda(A)^n = \lambda(A^{\otimes n})  
  \le \frac{ 2^{\rho n} |\mathcal{F}_n|+\sum_{i=1}^ns^{(n)}_i|\mu|^{-i}}{( 2^{\rho n} |\mathcal{F}_n|)^2},
\end{align}
where $s^{(n)}_i$ is the number of pairs $(\mathbf{u},\mathbf{v})\in \mathcal{F}_n\times \mathcal{F}_n$ such that 
$A^{\otimes n}(\mathbf{u},\mathbf{v})=|\mu|^{-i}$. 
We now give an upper bound for $s^{(n)}_i$ and the sum $\sum_{i=1}^ns^{(n)}_i|\mu|^{-i}$
through a simple counting argument. First let us introduce two functions which will simplify our derivations. Recall that $m$ is the number of vertices of $G$ and $e$ is its number of edges. For $1\le i\le n$, let 
\begin{align*}
  f(i) \dfn m^{n-i}(2e)^i{n\choose i} \mbox{\quad and\quad} g(i) \dfn m^{n-i}(2e)^i{n\choose i}|\mu|^{-i}.
\end{align*}

\begin{lemma}\label{fg_eq}
  \ \\[-5mm]
  \begin{enumerate}
    \item For any $1\le i\le n$,
      \begin{align*}
        s^{(n)}_i &\le  f(i), \\
        \sum_{i=1}^{n}s^{(n)}_i &\le  2^{\rho n}( 2^{\rho n}-1)|\mathcal{F}_n|.
      \end{align*}
    \item For any $1\le k\le n$,
          \begin{align*}
            \sum_{i=1}^{n}s^{(n)}_i|\mu|^{-i} 
            \le  2^{\rho n}( 2^{\rho n}-1)|\mathcal{F}_n||\mu|^{-k} +\sum_{i=1}^{k}g(i).
          \end{align*}
    \item For any $0\le q\le 1$,
      \begin{align*}
        \lim_{n\to\infty}\frac{1}{n}\log{f(qn)} 
        &= \log{(m(r+1))}-D\left(q \| \frac{r}{r+1}\right),\\
        \lim_{n\to\infty}\frac{1}{n}\log{g(qn)} 
        &= \log{\frac{m(r+|\mu|)}{|\mu|}}-D\left(q \| \frac{r}{r+|\mu|}\right).
      \end{align*}
    \item Write $C_\rho= C_{\rho}(G)$. Then for $0\le q\le \frac{r}{r+|\mu|}$,
          \begin{align*}
            &\limsup_{n\to\infty}\frac{1}{n}\log \sum_{i=1}^{n}s^{(n)}_i|\mu|^{-i} \\
            \le & \max\Big{\{}2\rho+C_\rho-q\log|\mu|, \\
                &\ \ \ \ \ \ \  \log{\frac{m(r+|\mu|)}{|\mu|}}-D\left(q \| \frac{r}{r+|\mu|}\right)\Big{\}}.
          \end{align*}
  \end{enumerate}
\end{lemma}
\begin{proof}
  \ \\[-5mm]
  \begin{enumerate}
    \item
    For two vertices $\mathbf{u}=(u_1,\dots,u_n)$ and $\mathbf{v}=(v_1,\dots,v_n)$ of $G^n$,
    the element $A^{\otimes n}(\mathbf{u},\mathbf{v})=|\mu|^{-i}$ if and only if there are $n-i$ pairs of
    coordinates such that $u_j=v_j$ and the other $i$ pairs are adjacent in graph $G$.
    So for $1\le i\le n$ we get
    $$
    s^{(n)}_i \le  m^{n-i}(2e)^i{n\choose i}=f(i).
    $$
    We see that $\frac{1}{2}\sum_{i=1}^{n}s^{(n)}_i$ is the number of edges that connect pairs of vertices in $\mathcal{F}_n$. By the definition of a $2^{\rho n}$-independent family, this number is upper bounded by $\frac{1}{2}  2^{\rho n}( 2^{\rho n}-1)|\mathcal{F}_n|$.
    \item
    We have
    \begin{align*}
      \sum_{i=1}^{n}s^{(n)}_i|\mu|^{-i} 
      &\le \sum_{i=1}^kg(i)+\sum_{k+1}^ns^{(n)}_i|\mu|^{-i}\\
      &\le \sum_{i=1}^kg(i)+|\mu|^{-k}\sum_{k+1}^ns^{(n)}_i\\
      &\le \sum_{i=1}^kg(i)+|\mu|^{-k} 2^{\rho n}( 2^{\rho n}-1)|\mathcal{F}_n|.
    \end{align*}
    \item
      Note that the first equality is a special case of the second with $|\mu|=1$. We have
    \begin{align*}
      &\lim_{n\to\infty}\frac{1}{n}\log{g(qn)}\\
      &= \lim_{n\to\infty}\frac{1}{n}\log{\left( m^{n-qn}(2e)^{qn}{n\choose qn}|\mu|^{-qn}\right)}\\
      &= (1-q)\log{m}+q\log{(2e)}+h(q)-q\log{|\mu|}\\
      &= \log{\frac{m(r+|\mu|)}{|\mu|}}-D\left(q \| \frac{r}{r+|\mu|}\right).
    \end{align*}
    \item 
      Set $k=qn$ in 2), and choose $\mathcal{F}_n$ to asymptotically achieve $C_\rho$. Then this argument follows from
    \begin{align*}
      &\lim_{n\to\infty}\frac{1}{n}\log( 2^{\rho n}( 2^{\rho n}-1)|\mathcal{F}_n||\mu|^{-qn})\\
      &= 2\rho+C_\rho-q\log|\mu|,
    \end{align*}
    and
    \begin{align*}
        & \lim_{n\to\infty}\frac{1}{n}\log\sum_{i=1}^{qn}g(i)\\ 
      &= \log{\frac{m(r+|\mu|)}{|\mu|}}-D\left(q \| \frac{r}{r+|\mu|}\right).      
    \end{align*}
\end{enumerate}
\end{proof}

We are now ready to state our bound. 
\begin{theorem}\label{bound_for_regular_graph}
  Let $G=(V,E)$ be a regular graph with $m$ vertices, $e$ edges and degree $r$.
  Let $\mu$ be its smallest eigenvalue.
  Then for any $\rho$ satisfying 
  \begin{align}\label{eq:rho_interval}
  \frac{1}{2}\log{\frac{r+|\mu|}{|\mu|}}<\rho<\log\frac{r+|\mu|}{|\mu|}+\frac{r}{r+|\mu|}\log{|\mu|},     
  \end{align}
  it holds that 
  \begin{align}\label{eq:bound}
     C_{\rho}(G) \le \log{m}-\rho-\frac{1}{2}D\left(p\,\|\,\frac{r}{r+|\mu|}\right),
  \end{align}
  where $0 < p < \frac{r}{r+|\mu|}$ is the unique solution of
  \begin{equation}\label{eq:rho_eq}
  \rho = \log\frac{r+|\mu|}{|\mu|}+p\log{|\mu|}-\frac{1}{2}D\left(p\,\|\,\frac{r}{r+|\mu|}\right).    
  \end{equation}
\end{theorem}
\begin{proof}
  Write $C_\rho= C_{\rho}(G)$, and let $\mathcal{F}_n$ asymptotically achieve $C_\rho$. Suppose $0<q<\frac{r}{r+|\mu|}$. 
  By claim 2) of Lemma~\ref{lemma:lambda}, inequality~\eqref{bound2} and claim 4) of Lemma~\ref{fg_eq}, we have
  \begin{align}\label{boundbound}
    \begin{split}
         -\log{\frac{m|\mu|}{r+|\mu|}} 
  &\le \limsup_{n\to\infty}\frac{1}{n}\log{\frac{ 2^{\rho n} |\mathcal{F}_n|+\sum_{i=1}^ns^{(n)}_i|\mu|^{-i}}{( 2^{\rho n} |\mathcal{F}_n|)^2}}\\
  &= \limsup_{n\to\infty}\frac{1}{n}\log\left({ 2^{\rho n} |\mathcal{F}_n|+\sum_{i=1}^ns^{(n)}_i|\mu|^{-i}}\right)-2\rho-2C_{\rho}\\
  &\le \max\Big\{C_{\rho}+\rho,\, 2\rho+C_\rho-q\log|\mu|,\,\log{\frac{m(r+|\mu|)}{|\mu|}}-D\left(q \| \frac{r}{r+|\mu|}\right)\Big\}-2\rho-2C_{\rho}.
\end{split}
  \end{align}
  Rearranging the terms in~\eqref{boundbound} we get
  \begin{align}\label{twoupperbounds}
    C_{\rho}\le\max\Big\{ \log{\frac{m|\mu|}{r+|\mu|}}-\rho, \log{\frac{m|\mu|}{r+|\mu|}}-q\log{|\mu|},\, \log{m}-\rho-\frac{1}{2}D\left(q\,\|\,\frac{r}{r+|\mu|}\right)\Big\}.
  \end{align}
  The difference of the last two terms in~\eqref{twoupperbounds} is
  \begin{align*}
    \Delta(q)\dfn &\Big(\log{\frac{m|\mu|}{r+|\mu|}}-q\log{|\mu|}\Big)-\Big(\log{m}-\rho-\frac{1}{2}D\left(q\,\|\,\frac{r}{r+|\mu|}\right)\Big) \\
    =& \rho - \Big( \log\frac{r+|\mu|}{|\mu|}+q\log{|\mu|}-\frac{1}{2}D\left(q\,\|\,\frac{r}{r+|\mu|}\right)\Big).
  \end{align*}
It is easy to check that $\Delta(q)$ is a continuous and strictly decreasing function of $q$ in the interval $[0,\frac{r}{r+|\mu|}]$. Morever, for any $\rho$ satisfying~\eqref{eq:rho_interval}, $\Delta(0)>0$  and $\Delta(r/(r+|\mu|))<0$. Hence there exists a unique value $p$ satistfying~\eqref{eq:rho_eq}, and the bound~\eqref{eq:bound} follows by setting $q=p$, which is optimal, in~\eqref{twoupperbounds}.

\end{proof}

\begin{example}\label{ex:pentagon}{\rm
    We apply Theorem~\ref{bound_for_regular_graph} to the cycle graph $C_5$ (pentagon), whose
    vertex set is $\{1,2,3,4,5\}$ and edge set is $\{12,23,34,45,51\}$. Then $r=2$ and
    $\mu = -(\sqrt{5}+1)/2$. The results are depicted
    in Figure~\ref{pentagon}. A lower bound can be obtained using~\eqref{lower_bound_rho} of Theorem~\ref{thm:bounds} with the following three independent families of
    $C_5^2$: 
    (i) $\{\{1\}\times\{1,2,3,4,5\},\{3\}\times\{1,2,3,4,5\},\{4\}\times\{1,2,3,4,5\}\}$;
    (ii) $\{\{4\}\times\{1,2\},\{1,2\}\times\{1,2\},\{1,2,3,4,5\}\times\{4\}$;
    (iii) $\{\{4,5\}\times\{5\},\{2\}\times\{1,5\},\{1,2\}\times\{3\},\{4\}\times\{2,3\}\}.$
    The other bounds are obtained using Proposition~\ref{lem:trivial_bound_rho}. 
}
\begin{figure}[!t]
    \centering
    \includegraphics[width=3.5 in]{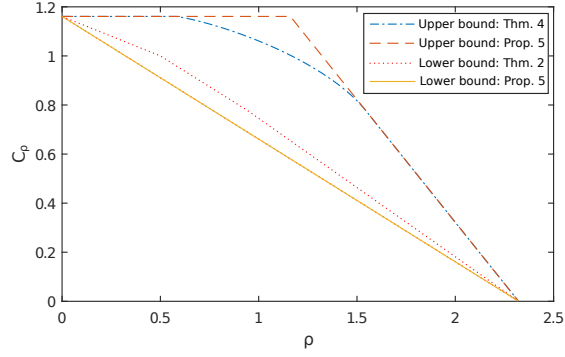} 
    \caption{Bounds on the $\rho$-capacity of the Pentagon (Example~\ref{ex:pentagon})}
    \label{pentagon}
\end{figure}
\end{example}

\section{Some Properties of the $\rho$-Capacity}\label{sec:properties}
In this section, we study the properties of the $\rho$-capacity function under two graph operations: the strong product and the disjoint union. To that end, we first observe the following three simple lemmas, whose proofs are relegated to the Appendix. 

\begin{lemma}\label{IndependentFamily}
  Let $G$ be a graph with $m$ vertices and $2\le k\le m$. Suppose $\mathcal{F}=\{V_i \mid 1\le i\le N\}$ is an independent family of $G$ such that $|V_i|\le k$ for $1\le i\le N$. Then 
  $$\sum_{i=1}^{N}|V_i|\le \min\{m, (k-1)(2\,\alpha_k(G)+1)\}.$$
\end{lemma}

\begin{lemma}\label{IndependenceNumber}
  Let $G=H_1+H_2+\cdots+H_n$ be the disjoint union of $n$ graphs $H_1,\ldots,H_n$ and $k\ge2$. Then
  $$\sum_{i=1}^{n}\alpha_{k}(H_i)\le \alpha_{k}(G) \le \min\left\{\frac{|V(G)|}{k}, \frac{k-1}{k}\sum_{i=1}^{n}(2\,\alpha_{k}(H_i)+1)\right\}.$$
\end{lemma}

\begin{lemma} \label{IndependenceNumberStrongProduct}
  Let $G$ be a graph. Then for any positive integer $k\le|V(G)|$, we have
  $\alpha_{k}(G) = \alpha_{km}(G\boxtimes K_{m}).$
\end{lemma}

\subsection{Strong Product} \label{subsec:strong_prod}
The following theorem provides a lower bound on the $\rho$-capacity of a strong product.
\begin{theorem}
  Let $G$ be a graph with $m_1$ vertices and $H$ a graph with $m_2$ vertices. Then
  \begin{align}\label{StrongProductInequality}
  C_\rho(G\boxtimes H) \geq \max_{\rho_1+\rho_2=\rho} C_{\rho_1}(G) + C_{\rho_2}(H)
  \end{align}
  where the max is taken over $0\le\rho_1\le\log{m_1}$ and $0\le\rho_2\le\log{m_2}$ such that $\rho_1+\rho_2=\rho$.
  \label{StrongProduct}
\end{theorem}
\begin{proof}
  By Lemma~\ref{StrongProductIndependenceNumberInequality} we get
$$
\alpha_{2^{(\rho_1+\rho_2)n}}((G\boxtimes H)^n) 
\ge \alpha_{2^{\rho_1 n}}(G^n)\cdot \alpha_{2^{\rho_2 n}}(H^{n}).
$$
The result immediately follows.  
\end{proof}

\begin{remark} {\rm We note that Theorem~\ref{StrongProduct} says that the function $C_{\rho}(G\boxtimes H)$ is lower bounded by the \textit{supremal convolution} of $C_{\rho}(G)$ and $C_{\rho}(H)$. Equivalently, it says that the hypograph of $C_\rho(G\boxtimes H)$ contains the Minkowsky sum of the hypographs of $C_\rho(G)$ and $C_\rho(H).$
}
\end{remark}

The following is a simple corollary of Theorem~\ref{StrongProduct}. 
\begin{corollary}
The concave conjugate of the $\rho$-capacity is subadditive with respect to the strong product, i.e., 
  $$
   C_{\star}(G\boxtimes H,\gamma)\le C_{\star}(G,\gamma)+ C_{\star}(H,\gamma).
  $$
\end{corollary}
\begin{proof}
  Let $m_1=|V(G_1)|$ and $m_2=|V(G_2)|$. Then
  \begin{align*}
     C_{\star}(G\boxtimes H,\gamma)
    &=\inf_{\rho\in[0,\log(m_1m_2)]}\gamma\rho-C_{\rho}(G\boxtimes H) \\
    &=\underset{\substack{\rho_1\in[0,\log{m_1}]\\\rho_2\in[0,\log{m_2}]}}{\inf}\gamma(\rho_1+\rho_2)-C_{\rho_1+\rho_2}(G\boxtimes H) \\
    &\le\underset{\substack{\rho_1\in[0,\log{m_1}]\\\rho_2\in[0,\log{m_2}]}}{\inf}\gamma(\rho_1+\rho_2)-C_{\rho_1}(G)-C_{\rho_2}(H) \\
    &=\left(\inf_{\rho_1\in[0,\log{m_1}]}\gamma\rho_1-C_{\rho_1}(G)\right) + \left(\inf_{\rho_2\in[0,\log{m_2}]}\gamma\rho_2-C_{\rho_2}(H)\right) \\
    &= C_{\star}(G,\gamma)+ C_{\star}(H,\gamma),
  \end{align*}
  where the inequality follows from Theorem~\ref{StrongProduct}.
\end{proof}

When $H$ is taken to be a complete graph in Theorem~\ref{StrongProduct}, then the lower bound~\eqref{StrongProductInequality} is attained.
\begin{theorem}\label{GXK_m}
  For any graph $G$, we have
  $$ C_{\rho}(G\boxtimes K_m) =
  \begin{cases}
    C(G) & \emph{if }0\le \rho < \log{m} \\
    C_{\rho-\log{m}}(G) & \emph{if }\log{m}\le\rho\le\log{|V(G\boxtimes K_{m})|}.
  \end{cases}
  $$
\end{theorem}
\begin{proof}
  First we can readily verify that $C(G\boxtimes K_m)=C(G)$. Since $(G\boxtimes K_m)^n\cong G^n\boxtimes K_{m^n}$, by Lemma~\ref{IndependenceNumberStrongProduct} we have
  $$\alpha_{2^{\rho n}}(G^n) = \alpha_{2^{\rho n}m^n}(G^n\boxtimes K_{m^n}).$$  
  It follows that 
  $$C_{\rho}(G\boxtimes K_m)=C_{\rho-\log{m}}(G) \text{  for  } \rho\ge\log{m}.$$
  For $0<\rho<\log{m}$, the result follows from the monotonically decreasing property of the $\rho$-capacity.
\end{proof}

\subsection{Disjoint Union} \label{subsec:disjoint_union}
We begin by finding the $\rho$-capacity of a union of two identical graphs, in terms of the $\rho$-capacity of a single copy. 
\begin{theorem}
  Let $G$ be a graph with $m$ vertices. Then
  \begin{align*}
    C_{\rho}(G+G) =
    \begin{cases}
      1+C_{\rho}(G) & \emph{if }0\le\rho<\log{m}\\
      1-\rho+\log{m} & \emph{if }\log{m}\le\rho\le 1+\log{m}.
    \end{cases}
  \end{align*}
  \label{DisjointUnion}
\end{theorem}
\begin{proof}
  It can be easily verified that $C(G+G)=1+C(G)$. Now we consider the case $0<\rho\le\log{m}$. From $(G+G)^{n} \cong 2^{n}G^{n}$ and Lemma~\ref{IndependenceNumber} we have
  $$2^n\cdot\alpha_{2^{\rho n}}(G^n) 
  \le \alpha_{2^{\rho n}}((G+G)^n)
  < 2^{n}\cdot\left(2\,\alpha_{2^{\rho n}}(G^n)+1\right).$$
It follows that $C_{\rho}(G+G)=1+C_{\rho}(G)$ for $0<\rho\le\log{m}$. The remaining case follows directly from 1) of Corollary~\ref{free-lunch-rate-and-packing-rate}.
\end{proof}

We now provide a lower bound on the $\rho$-capacity of a general disjoint union.
\begin{theorem} \label{DisjointUnion_2}
  Let $G$ be a graph with $m_1$ vertices and $H$ a graph with $m_2$ vertices, and let 
  $$\delta=(m_1\log{m_1}+m_2\log{m_2})/(m_1+m_2).$$
  Then
  \begin{align}\label{DisjointUnionInequality}
    C_{\rho}(G+H)
    \begin{cases}
      \ge \underset{p\rho_1+(1-p)\rho_2=\rho}{\max} h(p)+p\,C_{\rho_1}(G)+(1-p)C_{\rho_2}(H) & \emph{if } 0\le\rho<\delta \\
      =\log{(m_1+m_2)}-\rho & \emph{if }\delta\le\rho\le\log{(m_1+m_2)},
    \end{cases}
  \end{align}
  where the max is taken over $0\le p\le 1, 0\le \rho_1\le \log{m_1}$, and  $0\le\rho_2\le\log{m_2}$ such that $p\rho_1+(1-p)\rho_2=\rho$.
\end{theorem}
\begin{proof}
  By 1) of Corollary~\ref{free-lunch-rate-and-packing-rate} we have $C_{\rho}(G+H)=\log(m_1+m_2)-\rho$ if $\rho\ge\frac{1}{m_1+m_2}\sum_{i=1}^2m_i\log{m_i}=\delta$.
  By the properties of the strong product we have
  $$(G+H)^n \cong \sum_{i=0}^{n}{n\choose i}G^i\boxtimes H^{n-i}.$$
  Let $i=pn$. By Lemma~\ref{StrongProductIndependenceNumberInequality} we have
  \begin{align*}
  \alpha_{2^{(p\rho_1+(1-p)\rho_2)n}}((G+H)^n) \ge {n\choose pn}\cdot \alpha_{2^{\rho_1 pn}}(G^{pn})\cdot \alpha_{2^{\rho_2 (1-p)n}}(H^{(1-p)n}).
  \end{align*}
Now the first inequality follows directly. 
\end{proof}

If $H$ is taken to be a complete graph in Theorem~\ref{DisjointUnion_2}, then the lower bound~\eqref{DisjointUnionInequality} is attained.
\begin{theorem}\label{G+K_m}
  Let $G$ be a graph with $m_1$ vertices, and let 
  $$\delta=(m_1\log{m_1}+m_2\log{m_2})/(m_1+m_2).$$
  Then
  \begin{align*}
   C_{\rho}(G + K_{m_2}) =
   \begin{cases}
     \underset{p\rho_1+(1-p)\log{m_2} \ge \rho}{\max} h(p)+p\,C_{\rho_1}(G) &\emph{if } 0\le\rho<\delta \\
     \log{(m_1+m_2)}-\rho & \emph{if }\delta\le\rho\le\log{m_1+m_2}.
 \end{cases} 
  \end{align*}
  Here the max is taken over $0\le p\le 1$ and $0\le \rho_1\le \log{m_1}$ such that $p\rho_1+(1-p)\log{m_2} \ge \rho$.
\end{theorem}
\begin{proof}
  The second equality follows from 1) of Corollary~\ref{free-lunch-rate-and-packing-rate} directly. 
  We have
  $$(G+K_{m_2})^n \cong \sum_{i=0}^{n}{n\choose i}G^i\boxtimes K_{m_2^{n-i}}.$$   

  First, assume that $m_1=m_2$ and $0\le\rho<\delta=\log{m_1}$. Then by Lemma~\ref{IndependenceNumber} we have
  \begin{align}\label{LowerUpperBounds}
  \sum_{i=0}^n{n\choose i}\alpha_{2^{\rho n}}(G^i\boxtimes K_{m_2^{n-i}}) 
  \le \alpha_{2^{\rho n}}((G+K_{m_2})^n)
  < \sum_{i=0}^n{n\choose i}(2\,\alpha_{2^{\rho n}}(G^i\boxtimes K_{m_2^{n-i}})+1).
\end{align}
Since $\alpha_{2^{\rho n}}(G^i\boxtimes K_{m_2^{n-i}})\ge1$, we get
\begin{align}
  C_{\rho}(G+K_{m_2})
  =\lim_{n\rightarrow\infty}\frac{1}{n}\log\left({\sum_{i=0}^n{n\choose i}\alpha_{2^{\rho n}}(G^i\boxtimes K_{m_2^{n-i}})}\right).
  \label{G+K_mLimit}
\end{align}
  Let $i=pn$, and $\rho_1=(\rho-(1-p)\log{m_2})/p$. 
  Through a similar analysis as Theorem~\ref{GXK_m}, we have
  \begin{align}
  \lim_{n\rightarrow\infty}\frac{1}{n}\log{\alpha_{2^{\rho n}}(G^{pn}\boxtimes K_{m_2^{(1-p)n}})}
  = \begin{cases} 
      p\,C_{\rho_1}(G) & \text{if }\rho\ge (1-p)\log{m_2} \\
      p\,C_{0}(G) & \text{if }\rho <(1-p)\log{m_2}.
    \end{cases}
  \end{align}
      Combining this with~\eqref{G+K_mLimit} we get
  \begin{align*}
  C_{\rho}(G+K_{m_2})
  =\underset{p\rho_1+(1-p)\log{m_2}\ge\rho}{\max} h(p)+p\,C_{\rho_1}(G).  
\end{align*}

Secondly, assume that $m_1<m_2$. If $0\le\rho\le\log{m_1}$, then it can be proved similarly as above. Thus, we can assume that $\log{m_1}<\rho<\delta$. Let $N$ be the largest integer such that $|V(G^i\boxtimes K_{m_2^{n-i}})|=m_1^{N}m_2^{n-N}\ge 2^{\rho n}$. Set
\begin{align*}
  A(n) &= \sum_{i=0}^N{n\choose i}\alpha_{2^{\rho n}}(G^i\boxtimes K_{m_2^{n-i}}),\\ 
  B(n) &= \frac{1}{2^{\rho n}}\sum_{i=N}^n{n\choose i}m_1^im_2^{n-i}. 
\end{align*}
Then
\begin{align}\label{A(n)}
 A(n)\le\alpha_{2^{\rho n}}((G+K_{m_2})^n).
 \end{align}
By Lemma~\ref{IndependenceNumber} we have 
\begin{align}\label{A(n)3}
  \begin{split}
  \alpha_{2^{\rho n}}((G+K_{m_2})^n)
  &<\sum_{i=0}^{N}{n\choose i}\left(2\alpha_{2^{\rho n}}(G^i\boxtimes K_{m_2^{n-i}})+1\right)+2\left(\alpha_{2^{\rho n}}\left(\sum_{i=N}^n{n\choose i}G^i\boxtimes K_{m_2^{n-i}}\right)+1\right)\\
  &\le 2(A(n)+B(n))+\sum_{i=0}^{N}{n\choose i}+2.
\end{split}
\end{align}
Following a similar analysis as above, we can obtain
\begin{align}\label{A(n)2}
 \lim_{n\rightarrow\infty}\frac{1}{n}\log{A(n)}= \underset{p\rho_1+(1-p)\log{m_2}\ge\rho}{\max} h(p)+p\,C_{\rho_1}(G).
 \end{align}
Following the same method of Lemma~\ref{optimization_problems} of the Appendix, we can prove 
\begin{align}\label{A(n)4}
 \lim_{n\rightarrow\infty}\frac{1}{n}\log{B(n)}=h(q)
 \end{align}
where $q\in[0,1]$ satisfying that $q\log{m_1}+(1-q)\log{m_2}=\rho$. 
Combining~\eqref{A(n)} --~\eqref{A(n)4}, we get
\begin{align*}
  C_{\rho}(G)
  = \lim_{n\rightarrow\infty}\frac{1}{n}\log{A(n)}
  &= \lim_{n\rightarrow\infty}\frac{1}{n}\log{\left(2(A(n)+B(n))+\sum_{i=0}^{N}{n\choose i}+2\right)} \\
 &= \underset{p\rho_1+(1-p)\log{m_2}\ge\rho}{\max} h(p)+p\,C_{\rho_1}(G).
\end{align*}
The case $m_1>m_2$ can be proved similarly.
\end{proof}

\begin{remark}{\rm
    Let $G$ and $H$ be two graphs. Suppose that $\Theta(G)=\log{A}$ and $\Theta(H)=\log{B}$ for some $A>0$ and $B>0$. Shannon~\cite[Theorem 4]{Sh1956} proved that $\Theta(G+H)\ge \log(A+B)$ and $\Theta(G\boxtimes H)\ge \log{A}+\log{B}$, and that both bounds hold with equality if the vertex set of one of the two graphs, say $G$, can be covered by $\alpha(G)$ cliques. He also conjectured that the equalities hold in general, which has been disproved by Alon~\cite{Alon1998}. Theorems~\ref{DisjointUnion_2}--\ref{G+K_m} can be seen as 
    generalizations of~\cite[Theorem 4]{Sh1956} to the $\rho$-capacity setting. It is thus interesting to ask whether~\eqref{DisjointUnionInequality} holds with equality under Shannon's conditions. 
}
\end{remark}

\section{$\rho$-Capacity and Structural Properties}\label{sec:structural}
In this section, we discuss the connection between the $\rho$-capacity and some simple structural properties of the graph. First, we prove that a graph $G$ is not connected if and only if its packing point $\rho_*(G)<\log{|V(G)|}$. More explicitly, we show that the gap between $\rho_*(G)$ and $\log{|V(G)|}$ is equal to the Shannon entropy of the distribution induced by the sizes of the connected components of $G$. We then proceed to show that the free-lunch point and the packing point of a graph $G$ coincide if and only if $G=sK_{n}$ for some positive integers $s$ and $n$. Lastly, we show that two disjoint union of cliques are isomorphic if and only if their $\rho$-capacity functions are the same. 

\begin{theorem}\label{PackingPoint}
  Let $G$ be a graph with $m$ vertices. Suppose that $G$ has $s$ connected components of sizes $m_1,\dots,m_s$. Let $Q=\{m_1/m,\dots,m_s/m\}$. Then $\rho_{*}(G) = \log{m} - H(Q)$. In particular, $\rho_{*}(G)<\log{m}$ if and only if G is not connected. 
\end{theorem}
\begin{proof}
  Let $G$ be the confusion graph of some point-to-point channel $p(y|x)$. By Theorem~\ref{UpperBoundFromVanishingError} we have
  \begin{align*}
C_\rho(G) 
  &\leq \max_{U- X- Y\atop H(X|U) \ge \rho} I(Y;U) \\
  &=\max_{U- X- Y\atop H(X|U) \ge \rho} H(X) - I(U;X|Y)- H(X|U)\\
  &\leq \log{m}-\rho.    
  \end{align*}
  Thus, a necessary condition to achieve a sum-rate of $\log{m}$ is that $H(X|U)=\rho$, $H(X)=\log{m}$ (i.e., $X$ is uniform),  and $I(U;X|Y) = 0$ (i.e.,  $U-Y-X$ also forms a Markov chain). Hence we can lower bound the packing point by 
  \begin{align*}
    \rho_{*}(G) \geq \min_{\substack{U-X-Y\\U-Y-X\\H(X)=\log{m}}} H(X|U).
  \end{align*}
For every $y$, define $S_y=\{x\mid x\in \mathcal{X}, p(x|y)>0\}$. Clearly, $S_y$ is a clique in $G$. The two Markov chains imply that $p(u|x) = p(u|y)$ whenever $p(u,x,y) >0$. Hence for any $y$, the distribution $p(u|x)$ is the same for each $x\in S_y$. 

This immediately implies that if $G$ is connected then $p(u|x)$ does not depend on $x$ at all. Hence $U$ and $X$ are independent, and thus
$$\rho_{*}(G) \geq \min_{\substack{U:U-X-Y\\U-Y-X\\H(X)=\log{m}}} H(X|U) = H(X) = \log{m}.$$

Now assume that $G$ is not connected. From the above arguments it is clear that $p(u|x)$ does not change inside each connected component of $G$. In other words, we have the Markov chain $U- Z- X$ where $Z=g(X)$ is a random variable that returns the index of the connected component of $G$ that $X$ lies in. Then $H(X|U)\geq H(X|Z)$, with equality if and only if $Z$ and $U$ are one-to-one. Since we want to minimize $H(X|U)$, we can without loss of generality assume that $U=Z$. From
$H(X)=\log{m}$ we see that $X$ is uniform. It is easy to verify that the only way to achieve that is by setting $p(x|u)$ to be uniform inside the connected component associated with $u$. This yields 
$$\rho_{*}(G)\geq H(X|U) = \sum_{i=1}^{s} \frac{m_i\log{m_i} }{m}= \log{m} - H(Q).$$
On the other hand, we have $\rho_{*}(G)\le\log{m} - H(Q)$ by Corollary~\ref{free-lunch-rate-and-packing-rate}. This completes the proof.
\end{proof}

\begin{corollary}\label{cor:unique_disj}
  Let $G$ be a graph with $m$ vertices. Suppose that $C_{0}(G)=\log{s}$ for some positive integer $s$, and there is a unique way (up to permutations) of writing $m$ as a sum of positive integers $m=m_1+m_2+\dots+m_t$ such that $H(\frac{m_1}{m},\ldots,\frac{m_t}{m})=\log{m}-\rho_*(G)$. If $t=s$ then $G$  is a disjoint union of cliques of sizes $m_1,\ldots,m_s$.    
\end{corollary}
\begin{proof}
From Theorem~\ref{PackingPoint} and the uniqueness assumption it must be that $G$ has $s$ connected components of sizes $m_1,\ldots,m_s$. If even one connected component is not a clique then $C_0(G)\geq \log{(s+1)} > \log{s}$, concluding the proof.  
\end{proof}
\begin{corollary}\label{cor:prime_cliques}
Let $G$ and $H$ be two graphs with the same number of vertices. Suppose $G$ is a disjoint union of cliques of distinct prime sizes. Assume $C_0(G) = C_0(H)$ and $\rho_*(G) = \rho_*(H)$. Then $G\cong H$. 
\end{corollary}
\begin{proof}
Write $G = K_{p_1}+ K_{p_2} + \cdots + K_{p_s}$ where $p_1<p_2<\ldots<p_s$ are distinct primes. Suppose $H$ has $t$ connected components of sizes $m_1\leq m_2\leq\ldots\leq m_t$. Since the packing points of $G$ and $H$ coincide, by Theorem~\ref{PackingPoint} we have
\begin{align}\label{eq:ent_eq}
  H\left(\frac{p_1}{m},\ldots,\frac{p_s}{m}\right) = H\left(\frac{m_1}{m},\ldots,\frac{m_t}{m}\right).
\end{align}
We now show that this entropy equality implies that $s=t$ and $p_i=m_i\ (1\le i\le s)$, which by Corollary~\ref{cor:unique_disj} will prove our claim. 

From~\eqref{eq:ent_eq} we have
\begin{align*}
  \prod_{i=1}^s p_i^{p_i} = \prod_{j=1}^t m_j^{m_j}.
\end{align*}
Fix any $j$, and let $i$ be such that $p_i | m_j$. Then clearly $p_i^{p_i} | m_j^{m_j}$. Thus for any $j$ there exists a subset $S_j\subseteq [s]$ such that 
\begin{align}\label{eq:mp_prod}
m_j^{m_j} = \prod_{i\in S_j} p_i^{p_i}  
\end{align}
Moreover, $\{S_1,\ldots,S_t\}$ form a partition of $[s]$. Now, if $|S_j|=1$ for all $j$ then we are done. Suppose to the contrary there exists $j$ such that $|S_j| > 1$. Then~\eqref{eq:mp_prod} implies that $m_j<\sum_{i\in S_j} p_i$. Thus we have $m = \sum_{j=1}^tm_j < \sum_{i=1}^sp_i = m$, in contradiction. 

\end{proof}

\begin{theorem}\label{free_lunch=packing}
  Let $G$ be a graph with $m$ vertices. Then $G$ is the disjoint union of $s$ copies of a complete graph $K_n$, i.e., $G=sK_{n}$ for some positive integers $s$ and $n$, if and only if $\rho^*(G)=\rho_*(G)$.
\end{theorem}
\begin{proof}
  If $G=sK_{n}$, then we can easily verify that $\rho^*(G)=\rho_*(G)=\log{n}$. Now suppose that $\rho^*(G)=\rho_*(G)$. Assume that $G$ has $s$ connected components of sizes $m_1,m_2,\dots,m_s$. From Theorem~\ref{PackingPoint} we have
  $$\rho^*(G)=\rho_*(G)=\sum_{i=1}^{s} \frac{m_i\log{m_i}}{m}.$$
  On the other hand, we know $C_{0}(G)\ge\log{s}$. For simplicity write $\rho^*=\rho^*(G)$. By the definition of the free-lunch point, we get $C_{\rho^*}=C_{0}(G)\ge\log{s}$. Then
  \begin{align*}
    \log{m}=\rho^*+C_{\rho^*}(G)\ge\sum_{i=1}^{s}\frac{m_i\log{m_i}}{m}+\log{s}\ge\log{m}.
  \end{align*}
  The above inequality holds if and only if $m_1=m_2=\dots=m_s$. This proves the result.
\end{proof}

\begin{theorem} \label{DisjointUnionOfCliques}
  Let $G=K_{m_1}+K_{m_2}+\dots+K_{m_s}$ and $H=K_{n_1}+K_{n_2}+\dots+K_{n_t}$ be two disjoint union of cliques. Suppose that the functions $C_{\rho}(G)$ and $C_{\rho}(H)$ coincide. Then $G\cong H$.
\end{theorem}
\begin{proof}
  Since the functions $C_{\rho}(G)$ and $C_{\rho}(H)$ coincide, the graphs $G$ and $H$ have the same number of vertices, i.e., $m_1+\dots+m_s=n_1+\dots+n_t$. If $m_1=\dots=m_s$, then we get $G\cong H\cong sK_{m_1}$ by Theorem~\ref{free_lunch=packing}. 
  A similar proof applies for the case $n_1=\dots=n_t$.
  
  Now assume that $s\ge2$ and $m_i,1\le i\le s$ are not all equal and $n_i,1\le i\le t$ are not all equal. By Theorem~\ref{rho_capacity_disjoint_union_of_cliques} we have
  \begin{align*}
    \rho^*&=\rho^*(G)=\rho^*(H)=\frac{1}{s}\sum_{i=1}^{s}\log{m_i}=\frac{1}{t}\sum_{j=1}^{t}\log{n_j}, \\
    \rho_*& = \rho_*(G) = \rho_*(H) = \frac{1}{m}\sum_{i=1}^s m_i\log{m_i}=\frac{1}{m}\sum_{j=1}^t n_j\log{n_j},\\
    C_{\rho^*}(G)&=C_{\rho^*}(H)=\log{s}=\log{t}.
  \end{align*}
  Hence $s=t$. Fix $\tilde{\rho}\in[\rho^*,\rho_*]$. By~\eqref{rho_expression}, there exist  $\tilde{\beta},\tilde{\gamma}\in[0,1]$ such that
  $$\tilde{\rho}=\left(\sum_{i=1}^s m_i^{\tilde{\beta}}\log{m_i}\right)\Big/\sum_{i=1}^s m_i^{\tilde{\beta}} 
      =\left(\sum_{i=1}^s n_i^{\tilde{\gamma}}\log{n_i}\right)\Big/\sum_{i=1}^s n_i^{\tilde{\gamma}}.$$ 
      Then, by Theorem~\ref{rho_capacity_disjoint_union_of_cliques} and~\eqref{rho-function-1}, we have
  \begin{align*}
  -\tilde{\beta}&=C_{\tilde{\rho}}'(G)=C_{\tilde{\rho}}'(H)=-\tilde{\gamma},\\
  \log{\left(\sum_{i=1}^s m_i^{\tilde{\beta}}\right)}
  &=C_{\tilde{\rho}}(G)+\tilde{\beta}\tilde{\rho}
  =C_{\tilde{\rho}}(H)+\tilde{\gamma}\tilde{\rho}
  =\log{\left(\sum_{i=1}^s n_i^{\tilde{\beta}}\right)}.
\end{align*}
Therefore
\begin{align}\label{power_sum_equality}
  \sum_{i=1}^s m_i^{\beta}=\sum_{i=1}^s n_i^{\beta}\quad \text{for } 0\le\beta\le1.
\end{align}
Without loss of generality, assume that $m_1\le \dots\le m_s$ and $n_1\le \dots\le n_s$. Note that both sides of~\eqref{power_sum_equality} are analytic functions of $\beta$ over the whole complex plane. As they coincide in the interval $[0,1]$, they must be identical over the whole complex plane. Now letting $\beta\rightarrow\infty$ we get $m_s=\max_{i}m_i=\max_{i}n_i=n_s$. Applying this argument recursively, we conclude that $m_i=n_i$ for $1\le i\le s$.

\end{proof}

\begin{example}\label{ex:corner_points}{\rm
We note that the ``corner points'' of the $\rho$-capacity curve for a disjoint union of cliques do not always characterize the sizes of the cliques. Let $G=12K_{2}+6K_{8}$ and $H=4K_{1}+13K_{4}+K_{16}$. We see that the graphs have the same number of vertices $|V(G)|=|V(H)|=72$, the same Shannon capacity $C_{0}(G) = C_{0}(H) = \log{18},$ the same free-lunch point $\rho^*(G) = \rho^*(H) = \frac{5}{3},$ and the same packing point $\rho_*(G) = \rho_*(H) = \frac{7}{3}$, but they are clearly not isomorphic. 
}
\end{example}

\section{Open Problems}\label{sec:open}
Below we mention a few problems of interest. 
\begin{problem}
  The $\rho$-capacity of small graphs. 
  \begin{enumerate}[(i)]
  \item Find the $\rho$-capacity of all the graphs with up to $4$ vertices. The following four graphs remain unsolved: 
\begin{figure}[H]
    \centering
    \includegraphics[width=3.5 in]{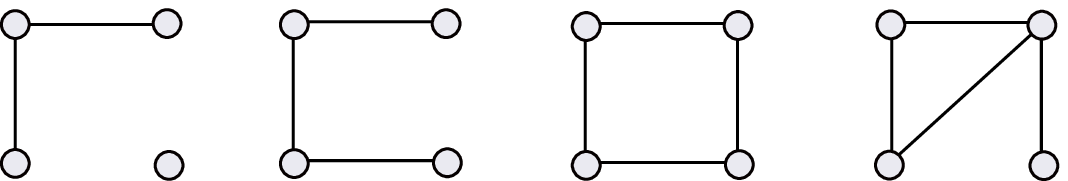} 
    \label{SmallGraphs}
\end{figure}
\item Find the $\rho$-capacity of the Pentagon $C_5$. 
  \end{enumerate}
\end{problem}

\begin{problem}
  Characterize the free-lunch point $\rho^*(G)$. Specifically, give a necessary and sufficient condition for $\rho^*(G)>0$. 
\end{problem}

\begin{problem}
  Let $G, H$ be two graphs with $C_\rho(G) = C_\rho(H)$. Do any of the following statements hold? 
  \begin{enumerate}
  \item If $G$ is a disjoint union of cliques then $G\cong H$.\footnote{Theorem~\ref{DisjointUnionOfCliques} and  Corollaries~\ref{cor:unique_disj} and~\ref{cor:prime_cliques} address special cases of this problem.} 
  \item If $G$ is a clique minus a clique then $G\cong H$.  
  \item If $E(G)\subseteq E(H)$ then $G\cong H$. 
  \item $G\cong H$. 
  \end{enumerate}
\end{problem}

\section*{Acknowledgment}
We are grateful to Lele Wang for many helpful discussions. We would also like to thank Alon Orlitsky for asking a question that ultimately led to this research, Noga Alon and Amit Weinstein for helpful discussions, and Young-Han Kim for pointing out the relevance of reference~\cite{Wi1990}. Finally, we are thankful to an annonymus referee for a very thorough and deep review of our work, and especially for noting a logical gap in the original proof of Theorem~\ref{PackingPoint}. 

\section*{Appendix}
\begin{lemma}\label{function-g}
  Let $m_1,\dots,m_s$ be $s$ positive integers, let $m=m_1+\dots+m_s$, and define a function $g:[-1,0]\rightarrow \RR$ by 
  $$g(\gamma)=-\log{\sum_{i=1}^s m_i^{-\gamma}}.$$ 
  Suppose that $s\ge2$ and $m_1,\dots,m_s$ are not all equal. Then
  \begin{enumerate}
    \item The function $g$ is differentiable on $[-1,0]$ and\,\footnote{Here $e$ is Euler's number, not the number of edges.}
      \begin{align*}
        g'(\gamma)&=\left(\sum_{i=1}^s m_i^{-\gamma}\log{m_i}\right)\Big{/}\sum_{i=1}^s m_i^{-\gamma},\\
        g''(\gamma)&=-(\log{e})\cdot\left(\sum_{1\le i<j\le s}(m_im_j)^{-\gamma}(\ln{m_i}-\ln{m_j})^2\right)\Big/\left(\sum_{i=1}^{s}m_i^{-\gamma}\right)^2.
      \end{align*}
    \item The function $g'$ is continuous and strictly monotonically decreasing on $[-1,0]$, and its image
      $$g'([-1,0]) = \left[\frac{1}{s}\sum_{i=1}^s \log{m_i}, \frac{1}{m}\sum_{i=1}^s m_i\log{m_i}\right].$$
  \end{enumerate}
\end{lemma}
\begin{proof}
  The result 1) follows from direct computation. Since $g''(\gamma)<0$ for $-1\le\gamma\le0$, we can verify 2) directly.
\end{proof}

\begin{lemma}\label{optimization_problems}
Let $m_1,\dots,m_s$ be $s$ positive integers, let $\rho$ be a nonnegative number satisfying that
\begin{align}\label{rho_interval_2}
  \frac{1}{s}\sum_{i=1}^s \log{m_i}\le\rho\le\frac{1}{m}\sum_{i=1}^s m_i\log{m_i},
\end{align}
and let
\begin{align*}
  A(n) &= \frac{1}{ 2^{\rho n}}\sum_{\substack{i_1+\cdots+i_s=n\\m_1^{i_1}\cdots m_s^{i_s}\le 2^{\rho n} }}{n\choose i_1,i_2,\ldots,i_s}m_1^{i_1}\cdots m_s^{i_2},\\
B(n) &= \sum_{\substack{i_1+\cdots+i_s=n\\m_1^{i_1}\cdots m_s^{i_s}\ge 2^{\rho n} }}{n\choose i_1,i_2,\ldots,i_s}.
\end{align*}
Suppose that $s\ge2$ and $m_1,\dots,m_s$ are not all equal. Then
  \begin{align*}
  \lim_{n\rightarrow\infty}\frac{1}{n}\log{A(n)}
    =\lim_{n\rightarrow\infty}\frac{1}{n}\log{B(n)}
    =\log{\left(\sum_{i=1}^s m_i^{\beta}\right)}-\beta\rho
  \end{align*}
  where $\beta\in[0,1]$ is the unique solution satisfying
   \begin{align*}
     \rho=\left(\sum_{i=1}^s m_i^{\beta}\log{m_i}\right)\Big{/}\sum_{i=1}^s m_i^{\beta}.
   \end{align*}
\end{lemma}
\begin{proof}
First, the number of tuples $(i_1,\ldots,i_s)$ such that $i_1+\cdots+i_s=n$ is at most $(n+1)^{s}$. For each tuple $(i_1,\ldots,i_s)$, let $P=(p_1,\ldots,p_s)=(i_1/n,\ldots,i_s/n)$. Then we have (see~\cite[Theorem 11.1.3]{CoTh2006})
  $$
  \frac{1}{(n+1)^{s}}\,2^{nH(P)}\le {n\choose i_1,i_2,\ldots,i_s} \le 2^{nH(P)}.
  $$ 
  Hence we can reduce the computation of $\lim_{n\rightarrow\infty}\frac{1}{n}\log{A(n)}$ and $\lim_{n\rightarrow\infty}\frac{1}{n}\log{B(n)}$
  to the following two maximization problems:
  \begin{align}\label{maximization_problem_1}
    \begin{split}
    \text{maximize\quad} & H(P)+\sum_{i=1}^sp_i\log{m_i}-\rho \\
    \text{subject to\quad} & 
    \sum_{i=1}^sp_i\log{m_i} \le \rho \\
    &\sum_{i=1}^sp_i=1 \\
    & p_i\ge0,\quad i=1,\ldots,s
  \end{split}
  \end{align}
  and
  \begin{align}\label{maximization_problem_2}
    \begin{split}
    \text{maximize\quad} & H(P) \\
    \text{subject to\quad} & 
    \sum_{i=1}^sp_i\log{m_i} \ge \rho \\
    &\sum_{i=1}^sp_i=1 \\
    & p_i\ge0,\quad i=1,\ldots,s.
  \end{split}
  \end{align}
  Now we will solve the maximization problem~\eqref{maximization_problem_1}. We first define the Lagrangian
  \begin{align*}
  &L(p_1,\ldots,p_s,\mu_0,\dots,\mu_{s},\lambda)\\
  =&H(P)+\sum_{i=1}^sp_i\log{m_i}-\rho-\mu_0\left(\sum_{i=1}^sp_i\log{m_i}-\rho\right)
  -\sum_{i=1}^{s}\mu_i(-p_{i})-\lambda\left(\sum_{i=1}^sp_i-1\right)\\
  =&H(P)+(1-\mu_0)\left(\sum_{i=1}^sp_i\log{m_i}-\rho\right)+\sum_{i=1}^{s}\mu_ip_{i}-\lambda\left(\sum_{i=1}^sp_i-1\right).
\end{align*}
By 2) of Lemma~\ref{function-g} and~\eqref{rho_interval_2}, there exists a unique $\tilde{\mu}_0\in[0,1]$ such that
\begin{align*}
 \rho = \left(\sum_{i=1}^s m_i^{1-\tilde{\mu}_0}\log{m_i}\right)\Big{/}\sum_{i=1}^s m_i^{1-\tilde{\mu}_0}.
\end{align*}
Let $\tilde{\mu}_i=0$ for $1\le i\le s$, and let 
$$ \tilde{p}_i=\frac{m_i^{1-\tilde{\mu}_0}}{\sum_{i=1}^{s}m_i^{1-\tilde{\mu}_0}},\quad i=1,\ldots,s, $$
and $\tilde{\lambda}=\log{(\sum_{i=1}^{s}m_{i}^{1-\tilde{\mu}_0})}-\log{e}$. Then we can verify that $\tilde{p}_1,\dots,\tilde{p}_s,\tilde{\mu}_0,\dots,\tilde{\mu}_s,\tilde{\lambda}$ satisfy the Karush--Kuhn--Tucker conditions (see~\cite[Section 5.5.3]{BoVa2004}). Therefore they are optimal and the maximum is
  \begin{align*}
    H(\tilde{p}_1,\ldots,\tilde{p}_s)=\log{\left(\sum_{i=1}^s m_i^{1-\tilde{\mu}_0}\right)}-(1-\tilde{\mu}_0)\rho. 
  \end{align*}
  Similarly, we can show that these $\tilde{p}_i, i=1,\dots,s$ are also optimal solutions for the maximization problem~\eqref{maximization_problem_2}. Now replacing $1-\tilde{\mu}_0$ with $\beta$ will give the result. 
\end{proof}

In the following, we provide the proofs of Lemmas~\ref{IndependentFamily}-\ref{IndependenceNumberStrongProduct} of Section~\ref{sec:properties}.

\begin{lemma1}
  Let $G$ be a graph with $m$ vertices and $2\le k\le m$. Suppose $\mathcal{F}=\{V_i \mid 1\le i\le N\}$ is an independent family of $G$ such that $|V_i|\le k$ for $1\le i\le N$. Then 
  $$\sum_{i=1}^{N}|V_i|\le \min\{m, (k-1)(2\,\alpha_k(G)+1)\}.$$
\end{lemma1}
\begin{proof}
  The inequality $\sum_{i=1}^{N}|V_i|\le m$ is obvious. Now, without loss of generality, we can assume that $|V_i|=k$ for $1\le i\le N_1$ and $|V_i|<k$ for $N_1<i\le N$. Then we can obtain a $k$-independent family from $\mathcal{F}$ as follows. First, set $U_i=V_i$ for $1\le i\le N_1$. Then we define $U_{N_1+1}=\cup_{i=N_1+1}^{N_1+l}V_{i}$, where $l$ is the smallest integer such that $\sum_{i=N_1+1}^{N_1+l}|V_i|\ge k$. Hence $k\le |U_{N_1+1}|\le 2(k-1)$. We continue in this way until the number of vertices contained in the remaining $V_i$ is less than $k$. Suppose the number of sets $U_i$ we get is equal to $M$. As these $U_i$ form a $k$-independent family, we have $M\le\alpha_{k}(G)$. It follows that
  $$ \sum_{i=1}^{N}|V_{i}| \le 2(k-1)\alpha_{k}(G)+k-1 = (k-1)(2\,\alpha_{k}(G)+1). $$
\end{proof}

\begin{lemma2}
  Let $G=H_1+H_2+\cdots+H_n$ be the disjoint union of $n$ graphs $H_1,\ldots,H_n$ and $k\ge2$. Then
  $$\sum_{i=1}^{n}\alpha_{k}(H_i)\le \alpha_{k}(G) \le \min\left\{\frac{|V(G)|}{k}, \frac{k-1}{k}\sum_{i=1}^{n}(2\,\alpha_{k}(H_i)+1)\right\}.$$
\end{lemma2}
\begin{proof}
  The first inequality can be easily verified, so we only deal with the second.
  Suppose $\mathcal{F}=\{V_j \mid 1\le j\le N\}$ is an $k$-independent family of $G$. Without loss of generality, we can assume that $|V_j|=k$ for all $j$. Now fix any $i$. Then $\{V_j\cap V(H_i)\mid 1\le j\le N\}$ is an independent family of $H_i$. By Lemma~\ref{IndependentFamily} we get
  \begin{align*}
    kN=\sum_{j=1}^{N}|V_j|=\sum_{i=1}^{n}\sum_{j=1}^{N}|V_j\cap V(H_i)|
    \le\min\left\{|V(G)|,(k-1)\sum_{i=1}^{n}\left(2\,\alpha_{k}(H_i)+1\right)\right\}.
  \end{align*}
  Now the result follows.
\end{proof}

\begin{lemma3}
  Let $G$ be a graph. Then for any positive integer $k\le|V(G)|$, we have
  $\alpha_{k}(G) = \alpha_{km}(G\boxtimes K_{m}).$
\end{lemma3}
\begin{proof}
  Let $\mathcal{F}=\{V_i\mid 1\le i\le N\}$ be a $km$-independent family of $G\boxtimes K_{m}$. For each $V_i$, if a vertex $(u,v)\in V_i$ where $u \in G,v\in K_{m}$, then without loss of generality we can assume that the set $u\times V(K_{m})$ is contained in $V_i$. Under this assumption, it is not hard to see that there is a one-to-one correspondence between the $k$-independent family in $G$ and the $km$-independent family in $G\boxtimes K_{m}$. The result follows easily from this observation.
\end{proof}

\bibliographystyle{IEEEtran}
\bibliography{ZeroErrorISIT}
\end{document}